\documentclass[11pt]{article}
\usepackage[left=1in, right=1in, top=1in, bottom=1in, margin=1in]{geometry}

\usepackage{algorithm}
\usepackage{tcolorbox}
\usepackage{amsmath,amssymb,xspace,graphicx,relsize,bm,xcolor,breqn,algpseudocode,multirow}

\usepackage{tikz}
\usetikzlibrary{quantikz2}
\usepackage{adjustbox}

\usepackage{tcolorbox}
\newcommand{\Exp}{\mathop{\mathbb{E}}}

\usepackage[margin=1in]{geometry}

\usepackage{amsthm}
\usepackage{dsfont}
\usepackage{array}
\usepackage{makecell}

\newcommand{\Sh}{\ensuremath{\mathsf{Stab}}}

\newcommand{\poly}{\ensuremath{\mathsf{poly}}}

\usepackage{enumerate}

\usepackage[pagebackref]{hyperref}
\usepackage[usestackEOL]{stackengine}
\usepackage{thm-restate,mathrsfs}
\usepackage{enumerate}
\usepackage{array}
\usepackage{parskip}
\def\01{\{0,1\}}

\hypersetup{
	colorlinks,
	linkcolor={blue!100!black},
	citecolor={red!100!black},
}

\newcommand{\be}{\begin{equation}}
\newcommand{\ee}{\end{equation}}
\newcommand{\ba}{\begin{array}}
\newcommand{\ea}{\end{array}}
\newcommand{\bea}{\begin{eqnarray}}
\newcommand{\eea}{\end{eqnarray}}

\usepackage{mathtools}

\DeclareMathOperator{\Tr}{Tr}
\newcommand{\ra}{\rangle}
\newcommand{\la}{\langle}

\newcommand{\calF}{{\cal F }}

\newcommand{\calS}{{\cal S }}

\newcommand{\calP}{{\cal P }}

\newcommand{\FF}{\mathbb{F}}

\newcommand{\gowers}[2]{\textsc{Gowers}\left({#1},{#2}\right)}

\declaretheorem[numberwithin=section]{theorem}

\declaretheorem[sibling=theorem]{problem}
\declaretheorem[sibling=theorem]{claim}

\declaretheorem[sibling=theorem]{corollary}

\theoremstyle{definition}

\declaretheorem[sibling=theorem]{fact}

\usepackage{thm-restate} 

\makeatletter
\def\widebreve{\mathpalette\wide@breve}
\def\wide@breve#1#2{\sbox\z@{$#1#2$}%
     \mathop{\vbox{\m@th\ialign{##\crcr
\kern0.08em\brevefill#1{0.8\wd\z@}\crcr\noalign{\nointerlineskip}%
                    $\hss#1#2\hss$\crcr}}}\nolimits}
\def\brevefill#1#2{$\m@th\sbox\tw@{$#1($}%
  \hss\resizebox{#2}{\wd\tw@}{\rotatebox[origin=c]{90}{\upshape(}}\hss$}
\makeatletter
\title{A note on polynomial-time tolerant testing stabilizer states}

\begin{document}

\author{
Srinivasan Arunachalam\\[2mm]
IBM Quantum\\
\small Almaden Research Center\\
\small \texttt{Srinivasan.Arunachalam@ibm.com}
\and
Sergey Bravyi\\[2mm]
IBM Quantum\\
\small IBM T. J. Watson Research Center\\
\small   Yorktown Heights, NY\\
\small \texttt{sbravyi@us.ibm.com}
\and
Arkopal Dutt\\[2mm]
IBM Quantum\\
\small Cambridge, MA\\
\small \texttt{arkopal@ibm.com}
}

\maketitle
\begin{abstract}
We show an improved inverse theorem for the Gowers-$3$ norm of $n$-qubit quantum states $\ket{\psi}$ which states that: for every $\gamma\geq 0$, if the $\gowers{\ket{\psi}}{3}^8 \geq \gamma$ then the stabilizer fidelity of $|\psi\ra$
is at least $\gamma^C$ for some constant $C>1$. This implies a constant-sample polynomial-time \emph{tolerant} testing algorithm for stabilizer states which accepts if an unknown state is $\varepsilon_1$-close to a stabilizer state in fidelity and rejects when $\ket{\psi}$ is  $\varepsilon_2 \leq \varepsilon_1^C$-far from all stabilizer states, promised one of them is the case. 
\end{abstract}

\section{Introduction}
We study the problem of tolerant testing stabilizer states which can be formalized as follows.
\begin{problem}[Tolerant testing]\label{prob:testing_stab}
Given access to copies of an unknown quantum state $\ket{\psi}$ with the promise $(i)$ $\ket{\psi}$ is $\varepsilon_1$-close to a stabilizer state in fidelity or $(ii)$ $\ket{\psi}$ is  $\varepsilon_2$-far from all stabilizer states, decide which is the case.
\end{problem}
Earlier works on this question considered the intolerant framework~\cite{gross2021schur,grewal2023improved} where $\varepsilon_1$ is either $1$ or a large constant, and $\varepsilon_2$ is decided accordingly.  Recently, \cite{ad2024tolerant} showed that there is a tolerant tester which could handle any $\varepsilon_1 \geq \tau$ and $\varepsilon_2 \leq 2^{-\poly(1/\tau)}$ with sample complexity $\poly(1/\varepsilon_1)$ but required a conjecture in additive combinatorics for arbitrary quantum states or a promise that the unknown state is a phase state (in which case $\varepsilon_2 \leq \varepsilon_1^{C'}$ for a constant $C' > 1$ could be handled). This was achieved by introducing Gowers norm for quantum states and showing an inverse theorem for the Gowers-$3$ norm of quantum states. 

\paragraph{Gowers norm and inverse theorems}
For any $n$-qubit quantum state $\ket{\psi}=\sum_{x \in \{0,1\}^n} f(x)\ket{x}$ where $f=(f(x))_x$ is an $\ell_2$-normed vector, we define its Gowers-$k$ norm~\cite{ad2024tolerant} as follows  
\begin{equation}
    \gowers{\ket{\psi}}{k} = 2^{n/2} \left[ \Exp_{x,h_1,h_2,\ldots,h_k \in \{0,1\}^n} \prod_{\omega \in \{0,1\}^k} C^{|\omega|} f(x + \omega \cdot h) \right]^{1/2^{k}},
\end{equation}
where $C^{|\omega|} f = f$ if $|\omega| := \sum_{j \in [k]} \omega_k$ is even and is $\overline{f}$ if $|\omega|$ is odd with $\overline{f}$ denoting the complex conjugate of $f$. We refer the interested reader to~\cite{ad2024tolerant} for more on Gowers norms. An inverse theorem for Gowers norms of quantum states asks the following: if a state $\ket{\psi}$ has high Gowers-$k$ norm, does it imply that $\ket{\psi}$ has high fidelity with a certain class of states? In \cite{ad2024tolerant}, this was answered affirmatively for the Gowers-$3$ norm.
\begin{theorem}[\cite{ad2024tolerant}]\label{thm:inversegowersstates}
    Let $\gamma\in [0,1]$. There exists a constant $C>1$ such that the following is true. Assuming a conjecture in additive combinatorics,\footnote{We refer the interested reader to~\cite[Conjecture~1.2]{ad2024tolerant} for the conjecture statement.} if $\ket{\psi}$ is an $n$-qubit quantum state such that $\gowers{\ket{\psi}}{3}^8 \geq \gamma$, then there is an $n$-qubit stabilizer state $\ket{\phi}$ such that $|\la \psi | \phi \ra|^2 \geq 2^{-\Omega(\gamma^{-C})}$.
\end{theorem}
In this note, we extend this result in two ways: $(i)$ we remove the requirement of the additive combinatorics conjecture (although we still believe it might be true and interesting independently), and $(ii)$ we improve the lower bound on stabilizer fidelity to $\poly(\gamma)$. In more detail, we show the theorem above can be extended to obtain the following inverse Gowers-$3$ theorem for states.
\begin{theorem}\label{thm:improved_inverse_gowers3}
Let $\gamma\in [0,1]$. If $\ket{\psi}$ is an $n$-qubit quantum state such that $\gowers{\ket{\psi}}{3}^8 \geq \gamma$, then there is an $n$-qubit stabilizer state $\ket{\phi}$ such that $|\la \psi | \phi \ra|^2 \geq \gamma^C$ for some constant $C>1$.
\end{theorem}
We reemphasize that the Gowers-$3$ norm takes values in $[0,1]$ and note that 
$\gowers{\ket{\psi}}{3}=1$
if and only $\ket{\psi}$ is a stabilizer state~(see Theorem~3.4 of \cite{ad2024tolerant}). The above theorem is also true provided $\Exp_{x \sim q_\Psi}[2^n p_\Psi (x)] \geq \gamma$ (where $p_\Psi$ is the characteristic distribution and $q_\Psi$ is the Weyl distribution~\cite{gross2021schur}) for a slightly different polynomial factor. The above inverse theorem implies a tolerant tester which can handle a broad regime of error parameters as stated as follows. Below, let $\Sh$ be the set of all $n$-qubit stabilizer states. 
\begin{theorem}\label{thm:improved_tolerant_tester}
There exists a constant $C>1$ such that the following is true. For any $\varepsilon_1>0$ there exists an algorithm that~given $\poly( 1/\varepsilon_1)$ copies of an $n$-qubit quantum state $\ket{\psi}$, decides if $\max_{\ket{\phi}\in \Sh}|\langle\phi |\psi\rangle|^2 \geq  \varepsilon_1$ or $\max_{\ket{\phi}\in \Sh}|\langle\phi |\psi\rangle|^2 \leq  \varepsilon_2$ for $\varepsilon_2 \leq \varepsilon_1^C$, using $O(n\cdot \poly(1/\varepsilon_1))$~gates.
\end{theorem}

\section{Proof of main result}
\subsection{Notation} 
In this section, we discuss the notation and prior results relevant for this note. Additionally, we denote $[m] = \{1,2,\ldots,m\}$ for an integer $m \geq 1$. We will work with the Weyl operators $W_x$ denoted by $(2n)$-bit strings $x = (a,b) \in \{0,1\}^{2n}$ given by
\begin{equation}
    W_x = i^{a \cdot b} \mathop{\bigotimes}_{j=1}^n X^{a_j} Z^{b_j}
    \label{eq:weyl_op}
\end{equation}
We denote the set of all $n$-qubit Weyl operators or $\{W_x\}_{x \in \FF_2^{2n}}$ as $\calP^n$. The Weyl operators are convenient to work with as they are orthogonal under the Hilbert-Schmidt inner product and allows us to write any quantum state in its Weyl (or Pauli) decomposition as
\begin{equation}
\ket{\psi}\bra{\psi} = \frac{1}{2^n} \sum_{x \in \{0,1\}^{2n}} \alpha_x W_x,
\label{eq:weyl_decomposition}
\end{equation}
where the coefficients $\alpha_x = \Tr(W_x \ket{\psi}\bra{\psi})$. Note that for a pure quantum state, we have that $\Tr(\Psi)=\Tr(\Psi^2)=1$ (denoting $\Psi = \ket{\psi}\bra{\psi}$) which implies $\alpha_{0^{2n}} = 1$ and $2^{-n} \sum_x \alpha_x^2 = 1$. Accordingly, we can define the characteristic distribution denoted as $p_\Psi$ and can be expressed as $p_\Psi(x) = 2^{-n} \alpha_x^2$. The samples of $p_\Psi$ can be accessed via Bell sampling~\cite{montanaro2017learning} on $\ket{\psi} \otimes \ket{\psi}^\star$. Another relevant distribution will be the Weyl distribution denoted by $q_\Psi$ which is equivalent to the convolution of $p_\Psi$ i.e., $q_\Psi(x) = \sum_{a \in \{0,1\}^{2n}} p_\Psi(a) p_\Psi(x+a)$. The Weyl distribution can be accessed via a routine called Bell difference sampling~\cite{gross2021schur} which only requires copies of $\ket{\psi}$. 

Lastly, we will work with isotropic subspaces $A \subset \FF_2^{2n}$ where all the Weyl operators corresponding to the $2n$-bit strings in $A$ commute with each other. Moreover, we say that an isotropic subspace $A$ is a Lagrangian subspace when it is of maximal size $2^n$ i.e., $|A| = 2^n$. The set of Weyl operators (up to a phase) corresponding to $A$ is then called a stabilizer group.

Finally, we denote $\calF_\calS(\ket{\psi}) = \max_{\ket{\phi}\in \Sh}|\langle\phi |\psi\rangle|^2$ as the stabilizer fidelity of $\ket{\psi}$. 
\begin{fact}[Proof of Theorem 3.3, \cite{gross2021schur}; Corollary 7.4, \cite{grewal2023improved}]
\label{fact:lower_bound_stabilizer_fidelity_pPsi_lagrangian_subspace}
For any $n$-qubit quantum state $\ket{\psi}$ and a Lagrangian subspace $T \subset \FF_2^{2n}$, we have $    \calF_\calS(\ket{\psi}) \geq \sum_{x \in T} p_\Psi(x)$.
\end{fact}

\subsection{Main lemma}
We now prove Theorems~\ref{thm:improved_inverse_gowers3} and \ref{thm:improved_tolerant_tester}.
Our starting point will be the following theorem which indicates that high Gowers-$3$ norm of quantum states implies the presence of structured subsets of $$
X = \{ x \in \{0,1\}^{2n} : 2^n p_\Psi(x) \geq \gamma/4\},
$$
which was useful in \cite{ad2024tolerant} enroute to their main~result. 
\begin{theorem}[Section~4.2,\,\cite{ad2024tolerant}]
\label{thm:existence_V}
Let $\gamma >0$. If $\gowers{\ket{\psi}}{3}^8 \geq \sqrt{\gamma}$, then there exists a subset $S' \subseteq  X$ with a small doubling constant $|S' + S'| \leq \poly(1/\gamma) \cdot |S'|$ and satisfying $\poly(\gamma) 2^n \leq |S'| \leq 4/\gamma \cdot  2^n$. There is a subgroup $V$ of size $|S'| \geq |V| \geq \poly(\gamma) |S'|$ with high intersection with $S'$, i.e.,
$$
\Exp_{y \in V}\left[ 2^n p_\Psi(y) \right] \geq \poly(\gamma).
$$
\end{theorem}
We now prove our main theorem in this note, which bounds the stabilizer covering of $V$.
\begin{theorem}\label{thm:stabilizer_covering_group}
Let $\ket{\psi}$ be an arbitrary pure quantum state and $V$ be a subgroup of $\{0,1\}^{2n}$ such that it has high intersection with the set $S' \subseteq X$, i.e., 
$
\Exp_{y \in V}\left[ 2^n p_\Psi(y) \right] \geq \poly(\gamma).
$
Then, $V$ can be covered by a union of $O(\poly(1/\gamma))$ many stabilizer subgroups $G_j\subseteq \01^{2n}$.
\end{theorem}

To do this, we need the following facts. In a slight abuse of notation, we will also use $V$ to denote the Weyl operators that correspond to the $2n$-bit strings in $V$. We use $X_i$ (or $Z_i$) to denote the Weyl operator with the single-qubit Pauli $X$ (or $Z$) acting on the $i$th qubit and trivially elsewhere. Moreover, we denote $\calP^m_Z=\{I,Z\}^{\otimes m}$. We write $I_\ell$ for the $\ell$-qubit identity operator.
\begin{fact}\label{fact:clifford_action_on_group}(\cite{fattal2004entanglement}) There exists $m+k \leq n$ and an $n$-qubit Clifford $U$~such~that
\begin{equation}
    U V U^\dagger = \la Z_1, X_1, \ldots, Z_k, X_k, Z_{k+1}, Z_{k+2}, \ldots, Z_{k+m} \ra.
\end{equation}
\end{fact}

Define the set $T := UVU^\dagger$. Using Fact~\ref{fact:clifford_action_on_group}, we have $T = UVU^\dagger = \calP^k \times \la Z_{k+1}, Z_{k+2}, \ldots, Z_{k+m}\rangle$.~Then
\begin{align*}
\frac{1}{|V|} \sum_{W_x \in \calP^k \times \la Z_{k+1},\ldots,Z_{k+m}\ra} 2^n p_{\tilde{\Psi}}(x) \geq \poly(\gamma)
\end{align*}
where the state $\tilde{\Psi}$ is the state obtained after application of the Clifford unitary $U$ i.e., $\tilde{\Psi} = U \Psi U^\dagger$. 


\begin{fact}
\label{eq:ub_sum1}
Let $\ell=n-k-m$.
We have the following
$$
\sum_{W_y \in \calP^k}
\sum_{W_z \in \calP^m_Z} \Tr\left( (W_y \otimes W_z \otimes I_\ell) \tilde{\Psi}\right)^2  \leq 2^{m+k} \,.
$$
\end{fact}
\begin{proof}
By definition, $\tilde{\Psi}$ is a projector onto a  normalized pure state. Thus
\[
 \Tr\left( (W_y \otimes W_z \otimes I_\ell) \tilde{\Psi}\right)^2  = \Tr\left( \tilde{\Psi} (W_y \otimes W_z \otimes I_\ell) \tilde{\Psi} (W_y \otimes W_z\otimes I_\ell)\right)
 \]
Since the Pauli group  is a unitary $1$-design, for any $k$-qubit operator $O$ one has $\sum_{W_y\in \calP^k} W_y O W_y = 2^k \Tr(O) \cdot I_k$.
As a consequence, 
\be
\sum_{W_y \in \calP^k} (W_y\otimes I_{m+\ell}) \tilde{\Psi}  (W_y\otimes I_{m+\ell}) =2^k I_k \otimes \rho
\ee
where  $\rho=\Tr_{1,\ldots,k} \tilde{\Psi}$ is an $(m+\ell)$-qubit density matrix obtained from $\tilde{\Psi}$ by tracing out the first $k$ qubits.
Accordingly, for any $W_z\in\calP^m_Z$ one has
\[
\sum_{W_y \in \calP^k} \Tr\left( (W_y \otimes W_z\otimes I_\ell) \tilde{\Psi}\right)^2 = 2^k \Tr\left( \rho (W_z\otimes I_\ell) \rho (W_z\otimes I_\ell)\right)\le
2^k\Tr(\rho)=
2^k.
\]
To get the last inequality we noted that  $\|(W_z \otimes I_\ell) \rho (W_z\otimes I_\ell)\|=\|\rho\|\le 1$.
We arrive at
\[
\sum_{W_y \in \calP^k}\sum_{W_z \in\calP^m_Z} \Tr\left( (W_y \otimes W_z\otimes I_\ell) \tilde{\Psi}\right)^2
= 2^k \sum_{W_z\in \calP^m_Z}  \Tr\left( \rho (W_z\otimes I_\ell) \rho (W_z\otimes I_\ell) \right) \le 2^{m+k}.
\]
\end{proof}

\begin{claim}\label{claim:small_k}
$
k \leq O(\log (1/\gamma)).
$
\end{claim}
\begin{proof}
We noted earlier that $\Exp_{y \in V} [2^n p_\Psi(y)] \geq \poly(\gamma)$.
Since 
$\Exp_{y \in V} [2^n p_\Psi(y)] = \Exp_{y \in UVU^\dag} [2^n p_{\tilde{\Psi}}(y)]$, one gets
\begin{align*}
    \poly(\gamma) |V| &\leq \sum_{W_x \in \calP^k \times \la Z_{k+1},\ldots,Z_{k+m}\ra} 2^n p_{\tilde{\Psi}}(x) = \sum_{W_y \in \calP^k, W_z \in \calP_Z^m} \Tr\left( (W_y \otimes W_z \otimes I_\ell) \tilde{\Psi}\right)^2 
    \leq 2^{k+m}
\end{align*}
where the second inequality uses Fact~\ref{eq:ub_sum1}. Since $|V|=2^{2k+m}$, this implies that $k=O(\log (1/\gamma))$.
\end{proof}

\begin{proof}[Proof of Theorem~\ref{thm:stabilizer_covering_group}.] We have previously noted that $T \equiv \calP^k \times \la Z_{k+1},\ldots,Z_{k+m} \ra$. Let us write the $4^k$ Paulis in $\calP^k = \{\tau_1,\tau_2,\ldots,\tau_{4^k}\}$. 
We choose $\tilde{S}_a = \la \tau_a \otimes I_{n-k} \ra \times \la Z_{k+1}, \ldots Z_{k+m} \ra \subseteq \calP^n$ for each $a \in [4^k]$. Moreover, we can extend $\tilde{S}_a$ to a stabilizer group $S_a$ on $n$-qubits. We then have 
\begin{equation}
    V \subseteq \bigcup \limits_{a=1}^{M} U^\dagger S_a U,
\end{equation}
where $M = 4^k$. Using Claim~\ref{claim:small_k}, we have $k \leq O(\log(1/\gamma))$ implying have that  $M=\poly(1/\gamma)$.~Moreover, the action of a Clifford $U$ on a stabilizer group $S_a$ is to take it to another stabilizer~group.
\end{proof}

\subsection{Putting everything together.} 
\paragraph{Inverse theorem.} Using Theorem~\ref{thm:stabilizer_covering_group}, we can cover $V$ with a disjoint union of $M = \poly(1/\gamma)$ many Lagrangian subspaces $\{S_a\}_{a \in [M]}$ (also called stabilizer groups). Thus, there exists a stabilizer group at $i^* \in [M]$ for which
$$
\frac{1}{|V|}\sum_{y\in V\cap S_{i^*}} 2^n p_\psi(y)\geq \poly(\gamma).
$$
We can now conclude by using Fact~\ref{fact:lower_bound_stabilizer_fidelity_pPsi_lagrangian_subspace} and considering the stabilizer group $S_{i^*}$ to show
\begin{align*}
    \calF_\calS(\ket{\psi}) \geq \sum_{y \in S_{i^\star}} p_\Psi(y) \geq \sum_{y \in S_{i^*} \cap V} p_\psi(y) \geq \poly(\gamma).
\end{align*}
This completes the proof of Theorem~\ref{thm:improved_inverse_gowers3}. The following corollary is immediately true then as well.
\begin{corollary}
\label{thm:inversegowersstates_weyl_expectations}
Let $\gamma\in [0,1]$. There exists a constant $C'>1$ such that the following is true. If $\ket{\psi}$ is an $n$-qubit quantum state such that  $\Exp \limits_{x \sim q_\Psi}\left[|\la \psi | W_x | \psi \ra|^2 \right] \geq \gamma$, then there exists an $n$-qubit stabilizer state $\ket{\phi}$ such that 
$
|\la \psi | \phi \ra|^2 \geq \gamma^{C'}.
$
\end{corollary}

\paragraph{Tolerant tester.} Our Theorem~\ref{thm:improved_tolerant_tester} follows similar to \cite[Theorem~1.3]{ad2024tolerant}: the  algorithm simply takes $O(1/\delta^2)$ copies of $\ket{\psi}$ and estimates $\Exp \limits_{x \sim q_\Psi}\left[|\la \psi | W_x | \psi \ra|^2 \right]$ up to additive error $\delta/2$. In the \emph{no} instance, our inverse Gowers theorem implies $\Exp \limits_{x \sim q_\Psi}\left[|\la \psi | W_x | \psi \ra|^2 \right]\leq \poly(\varepsilon_2)$ and in the \emph{yes} instance, we have $\Exp \limits_{x \sim q_\Psi}\left[|\la \psi | W_x | \psi \ra|^2 \right]\geq \varepsilon_1^6$. Setting $\delta=\poly(1/\varepsilon_1)$ appropriately, the tester can distinguish which is the case. The overall gate complexity is a factor-$n$ more than the sample complexity.

\subsection{Discussion}
\paragraph{Comparison to agnostic learning.} Recently, Chen et al.~\cite{chen2024stabilizer} considered agnostic learning stabilizer states (which is harder than tolerant testing). Here, an algorithm is given copies of $\ket{\psi}$ promised the stabilizer fidelity is at least $\tau$, needs to output a stabilizer state $\ket{\phi} \in \Sh$ which has a fidelity with $\ket{\psi}$ at least $|\la \phi | \psi \ra|^2 \geq \tau-\varepsilon$. Their agnostic learning algorithm uses $O(n\cdot (1/\tau)^{\log (1/\tau)})$ copies of $\ket{\psi}$ and time $O(n^3/\varepsilon^2\cdot (1/\tau)^{\log (1/\tau)})$. They also show a \emph{conditional} hardness  result showing that one cannot hope for a polynomial-time algorithm for $\tau=1/\poly(n)$ by assuming that a  non-standard variant of learning parity with noise is hard.  Interestingly, our note shows that if the goal was to consider the weaker task of \emph{testing} if $\ket{\psi}$ is $\tau$-close to a stabilizer state or $\poly(\tau)$-far, then one can obtain a \emph{constant-sample} (i.e., independent of $n$) and \emph{polynomial-time}~tester.

\paragraph{Acknowledgements.} Recently, Jop Briet and Jonas Helsen reached out to S.A.~and A.D.~with a proof sketch on how they might get a tolerant tester that can handle polynomial gaps building on \cite{ad2024tolerant}. Their idea was to  study the Lovasz-theta number of the graph corresponding to the subgroup $V$ and its graph products. Upon this discussion, S.A.~and A.D.~realized that they could achieve such a tolerant tester using a note shared by Sergey Bravyi much earlier. We are very grateful for this email exchange which motivated us to improve~\cite{ad2024tolerant} to obtain this note. We thank them for coordinating their submission~\cite{helsen2024} with ours on arXiv.


\bibliographystyle{alpha}
\bibliography{references}

\end{document}